\documentclass[12pt]{amsproc} \usepackage{amsmath,amsfonts,amsthm,mathrsfs,enumerate} \theoremstyle{plain} \newtheorem{lem}{Lemma} \newtheorem{thm}{Theorem} \newtheorem{cor}{Corollary}  \theoremstyle{definition} \newtheorem{define}{Definition} \theoremstyle{remark} \newtheorem{ex}{Example}  \DeclareMathOperator{\cl}{cl} \DeclareMathOperator{\interior}{int} \title[Existence of continuous euclidean embeddings]{Existence of continuous euclidean embeddings for a weak class of orders} \author{Lawrence Carr} \address{Department of Economics, Princeton University, Princeton, NJ 08544} \thanks{This work was funded by an undergraduate research grant from the Princeton University Department of Economics.} \date{August 2, 2015} \begin{document} \begin{abstract} We prove that if $X$ is a topological space that admits Debreu's classical utility theorem (eg.\ $X$ is separable and connected, second countable, etc.), then order relations on $X$ satisfying milder completeness conditions can be continuously embedded in $\mathbb R^I$ for $I$ some index set. In the particular case where $X$ is a compact metric space, this closes a conjecture of Nishimura \& Ok (2015). We also show that when $\mathbb R^I$ is given a non-standard partial order coinciding with Pareto improvement, the analogous embedding theorem fails to hold in the continuous case. \end{abstract} \maketitle \section{Existence of the embedding} Following Nishimura and Ok~\cite{ok}, we will be principally interested in the conditions under which a binary relation $P$ on a topological space $X$ admits a (continuous) Euclidean embedding in the sense that there exists a collection $\mathcal{V}$ of (continuous) maps $v:X\to\mathbb R$ such that \begin{align} (x,y)\in P\quad\iff\quad \exists v\in\mathcal{V},v(x)\geq v(y) \end{align} for every $x,y\in X$. When stated in this manner, the definition clarifies the interest in this construction within the study of choice theory. In the economic setting, $P$ represents an agent's preferences for goods and can be described by the family $\mathcal V$ of utility functions~\cite{evren}. In the language of mathematics, $\mathcal V=\{v_\alpha\}_{\alpha\in I}$ can be thought of as an embedding\footnote{Since the underlying set will always be $X$, we will not be pedantic about writing relations as a pair $(X,R)$.} of $(X,R)$ into $(\mathbb R^I,\geq)$, preserving both the order and topological structures (where $\mathbb R^I$ is given its usual product topology and the order $(x_\alpha)_{\alpha\in I}\geq(y_\alpha)_{\alpha\in I}\iff \exists\alpha\in I,x_\alpha\geq y_\alpha$). In order to prove that such embeddings exist, we must be more specific about the structure of the binary relation $R$. \begin{define} We always interpret a relation $R$ over a set $X$ as a subset of $X\times X$. If $X$ has a topology, then $R$ is automatically given the subspace topology. We say that a relation $R$ is \textit{continuous} if it is topologically closed relative to $X\times X$. We write $R^*$ for the reflection of $R$ across $\Delta_X$ (ie.\ the dual of $R$ viewed as a category), where $\Delta_X$ denotes the diagonal of $X\times X$. We write $x\sim y$ to mean $(x,y)\notin R\cup R^*$. A \textit{partial order} is a reflexive ($\Delta_X\subset R$), antisymmetric $(R\cap R^*\subset\Delta_X$), and transitive ($(x,y)\in R\land (y,z)\in R\implies (x,z)\in R$) relation. A \textit{linear order} is a complete ($R\cup R^*=X\times X$) partial order. \end{define} Throughout this paper it will be convenient to reference both the ``weak'' and ``strong'' versions of a relation which will generally be denoted by $P$ and $Q$ respectively; for instance the usual order on $\mathbb R$ can be thought of either as $P=\{(x,y)\in\mathbb R\times\mathbb R:x\geq y\}$ or $Q=\{(x,y)\in\mathbb R\times\mathbb R:x>y\}$. The choice is immaterial thanks to the duality $P=(Q^*)^c$, so we will freely make use of both notions. The weak class of orders we shall consider (where a higher dimensional embedding is necessary) is characterized by asymmetry ($\Delta_X\cap Q=\emptyset$) and transitivity of $Q$ or, dually, completeness and negative transitivity ($P^*$ is transitive) of $P$. The stronger class (where one-dimensional embeddings exist) has asymmetry and negative transitivity of $Q$ or completeness and transitivity of $P$. It is easy to verify that the strong conditions imply the weak conditions. \begin{lem}[Dushnik \& Miller \cite{1941}]\label{decomp} Let $P$ be a partial order on an arbitrary set $X$. Then there exists a collection $\{P_\alpha\}_{\alpha\in I}$ of linear orders on $X$ which realize $P$, ie. \begin{align*} P=\bigcap_{\alpha\in I}P_\alpha. \end{align*} Furthermore, there is such a realization of $P$ with the additional feature that whenever $x\sim y$, there exists an $\alpha$ in $I$ such that $(x,y)\in P_\alpha$. \end{lem} Observe that for an asymmetric and transitive relation $Q$ on a set $X$, the relation $P=Q\cup\Delta_X$ is a partial order so Lemma \ref{decomp} provides a collection $\{P_\alpha\}_{\alpha\in I}$ of linear orders realizing $P$, and thus $\{Q_\alpha\}_{\alpha\in I}$ where $Q_\alpha=P_\alpha\setminus\Delta_X$ is a family of asymmetric negatively transitive relations which realize $Q$. For $R$ a relation on a topological space $X$, denote by $\interior R$ and $\cl R$ the relations characterized by, respectively, its topological interior and closure relative to the product topology of $X\times X$. \begin{lem}\label{closed} If $Q$ is asymmetric and negatively-transitive, so is its interior. \end{lem} \begin{proof} It is easy to verify that $Q^c$ must be complete and transitive. $\cl(Q^c)$ can be thought of as the collection of limits of nets in $Q^c$ which converge in $X\times X$. Take $x,y,z\in X$ such that $(x,y)\in\cl(Q^c)$ and $(y,z)\in\cl(Q^c)$ and nets $(p_\alpha)_{\alpha\in A}=(x_\alpha,y_\alpha)_{\alpha\in A}$ and $(q_\beta)_{\beta\in B}(y_\beta',z_\beta')_{\beta\in B}$ that converge in $X\times X$ to $(x,y)$ and $(y,z)$ respectively. Let $U\subset X\times X$ be an open set around $r=(x,z)$. Then the product topology provides open sets $U_x,U_z\subset X$ with $U\supset U_x\times U_z$. Let $V\subset X$ be an arbitrary open set around $y$. Since $p_\alpha=(x_\alpha,y_\alpha)\to(x,y)$, $p_\alpha$ is eventually in $U_x\times V$, so $x_\alpha$ is eventually in $U_x$. By the same argument, $z_\beta'$ is eventually in $U_z$. One can confirm that $A\times B$ is a directed set, where $(p_\alpha,q_\beta)\geq(p_{\alpha'},q_{\beta'})$ if and only if $p_\alpha\geq p_{\alpha'}$ and $q_\beta\geq q_{\beta'}$. It can then direct the net $(r_\gamma)_{\gamma\in\Gamma}$, where $\Gamma=A\times B$ and $r_\gamma=(p_\alpha,q_\beta)\mapsto(x_\alpha,z_\beta')$. By construction $r_\gamma\to r=(x,z)$. Thus $(x,z)\in\cl(Q^c)$ and this proves transitivity. Completeness is trivial because $Q^c$ is contained in its closure. Thus $\interior(Q)=\cl(Q^c)^c$ is asymmetric and negatively-transitive. \end{proof} This lemma puts us in position to state and prove the main theorem. \begin{define} A topological space $X$ is a \textit{Debreu space} if every complete, transitive, continuous order can be continuously embedded in $(\mathbb R,\geq)$. \end{define} Several sufficient conditions for $X$ to be a Debreu space are known in the literature. For example: \begin{itemize} \item separability and connectivity~\cite{eilenberg} \item second-countability~\cite{debreu} \item separability and local-connectedness~\cite{candeal}. \end{itemize} \begin{thm}\label{conj} Let $P$ be a continuous binary relation on a Debreu space $X$. Then $P$ is complete and negatively transitive if and only if it is continuously embeddable in $(\mathbb R^I,\geq)$. \end{thm} \begin{proof} The proof of the reverse statement is straightforward. For the forward statement, Lemma \ref{decomp} and the note following it provide a collection of asymmetric and negatively-transitive relations $\{Q_\alpha\}_{\alpha\in I}$ such that $Q=\bigcap_{\alpha\in I}Q_\alpha$. It is a basic fact of topology that $\interior Q=\interior\bigcap_{\alpha\in I}Q_\alpha\subset\bigcap_{\alpha\in I}\interior Q_\alpha$. Since $P$ is continuous, $Q$ is open in $X\times X$ and $Q=\interior Q$. Since $\interior Q_\alpha\subset Q_\alpha$, the reverse containment also holds. Thus \begin{align*} Q=\bigcap_{\alpha\in I}\interior Q_\alpha. \end{align*} By Lemma \ref{closed}, each $\interior Q_\alpha$ for $\alpha\in I$ is asymmetric and negatively-transitive on $X$. Thus, for each, we can apply Debreu's theorem to find a continuous $v_\alpha:X\to\mathbb R$ such that $\interior Q_\alpha=\{(x,y)\in X\times X:v_\alpha(x)>v_\alpha(y)\}$, which proves the theorem when combined with the display equation. \end{proof} Hence we can prove a conjecture of Nishimura \& Ok~\cite{ok}. \begin{cor} Let $P$ be a continuous binary relation on a compact metric space $X$. Then $P$ is complete and negatively-transitive if and only if it is continuously embeddable in $\mathbb R^I$. \end{cor} \begin{proof} This follows immediately from Theorem \ref{conj} and Debreu's original version of the theorem~\cite{debreu}, as separable metric spaces have countable base. \end{proof} \section{Some implications of the main theorem} \begin{cor}\label{tauorder} For $P$ a partial order on a set $X$ and $\tau$ any topology on $X$, let the $\tau$-order dimension $d_\tau(P)$ be the cardinality of the minimal realization of $P$ by linear orders which are open in the topology $\tau$ on $X$. Then \begin{align} d_\tau(P)=d(P). \end{align} \end{cor} This corollary may simplify the problem of finding the dimension of an order on a set which admits a natural topology. In particular, the non-continuous dimension is equal to any continuous dimension. The following example of a semiorder is related to that first noted by Luce~\cite{luce} as a situation in which usual utility theory is inadequate. Our main theorem resolves the issue with a continuous order embedding that we can explicitly write down. \begin{ex}\label{sugar} An agent strictly prefers the larger of two quantities between which he can distinguish, but he can distinguish only between quantities which differ by an amount greater than some fixed $\epsilon>0$. That is to say $X=\mathbb{R}$ and \begin{align*} (x,y)\in P\quad\iff\quad x+\epsilon\geq y. \end{align*} It is easy to see that although $P$ is not transitive, these preferences do satisfy the conditions of Theorem \ref{conj}. Indeed, it is easy to verify that $P$ has the continuous order embedding \begin{align*} \mathcal{V}=\{f(\frac{x-\alpha}{\epsilon}):\alpha\in\mathbb{R}\} \end{align*} where $f(x)=x+(1-x^2)\chi(x)$ and $\chi$ is the indicator function of the interval $(-1,1)$. \end{ex} \begin{cor} If $P$ is a complete, negatively-transitive, continuous relation on $X$ which is compact (connected), then $P$ admits a Hasse diagram in which the collection of points is compact (connected) in $\mathbb{R}^2$. \end{cor} \begin{proof} Theorem \ref{conj} gives us a continuous map $f:X\to\mathbb{R}^{d(P)}$ such that $(x,y)\in Q$ if and only if $f(x)>f(y)$, where we compare vectors coordinate-wise. Let $\varphi:\mathbb{R}^{d(P)}\to\mathbb{R}^2$ be the projection onto a 2-plane through the identity line in $\mathbb{R}^{d(P)}$. Clearly $\varphi\circ f$ is continuous so it preserves compactness and connectedness, and maps $X$ to the appropriate Hasse diagram. \end{proof} It follows from Theorem 4.1 in \cite{1941} that for any cardinality $\kappa$, one can construct a complete, negatively-transitive order $P$ that cannot be non-continuously embedded in $(\mathbb R^I,\geq)$ for any $|I|<\kappa$. We might expect, however, some statement of minimality for \textit{continuous} representations, especially when $X$ is assumed to be separable. It is evident from Example~\ref{sugar} that not every uncountable index set $I$ such that $Q$ is embeddable in $\mathbb R^I$ has a countable subset such that the same embedding holds, as each $v\in\mathcal{V}$ contributes a unique point on the line $y=x+\epsilon$. One might still hope that there is a different family $\mathcal V$ which permits an embedding in $\mathbb R^\omega$. Alas, despite its simplicity, Example~\ref{sugar} shows us that this need not be the case. \begin{thm} The relation in Example~\ref{sugar} cannot be continuously embedded in $(\mathbb R^I,\geq)$ for any countable $I$. \end{thm} \begin{proof} Suppose for the sake of contradiction that $\mathcal V$ is a countable family achieving such an embedding. Then the collection of sets $\{(x,y)\in\mathbb{R}^2:v(x)\leq v(y)\}$ for $v\in\mathcal{V}$ is a cover of the line $x=y+\epsilon$. We can restate this as for any real $x$, there is a $v\in\mathcal{V}$ such that $v(x)\leq v(x-\epsilon)$ which implies \begin{align*} \emptyset=\bigcap_{v\in\mathcal{V}}\{x\in\mathbb{R}:v(x)> v(x-\epsilon)\}\supset\bigcap_{v\in\mathcal{V}}\{x\in F:v(x)> v(x-\epsilon)\} \end{align*} for any closed $F\subset\mathbb{R}$. Since each $v$ is assumed continuous, each of the sets in the intersection is open. It follows from the Baire category theorem that $F$ is a Baire space so $\{x\in F:v(x)> v(x-\epsilon)\}$ is not dense in $F$ for any non-empty closed $F$ or $v\in\mathcal{V}$. In particular, each such set is not dense in its closure, which is a contradiction. \end{proof} Moreover, it follows from Corollary~\ref{tauorder} that the order dimension of $Q$ in the sense of \cite{1941} is the continuum. \section{The question of Pareto embeddings} The notion of continuous order embedding in $\mathbb R^I$ is sensible formally as it is equivalent to the existence of a continuous monotonic function from $X$ into $(\mathbb{R}^{I},\geq)$. It is also mathematically convenient because if $\mathcal{V}$ realizes $Q=(P^*)^c$, then \begin{align*} P=\bigcup_{v\in\mathcal{V}}\{v(x)\geq v(y)\}\quad\text{and}\quad Q=\bigcap_{v\in\mathcal{V}}\{v(x)>v(y)\}. \end{align*} However, we argue as follows that this definition lacks the desired economic interpretation. One would expect that a Euclidean embedding has the effect of decomposing the agent's preferences---which are incomplete\footnote{That is to say that the strong relation is incomplete, or equivalently the weak relation is intransitive.} because, perhaps, he is considering several factors---into subdecisions which are total orders. This can be seen explicitly in a pair of examples. First consider the problem of the social planner who strictly prefers one allocation to another if the first is a Pareto improvement over the other. Such preferences are clearly incomplete with multiple agents because, if $x$ is the status quo and $y$ is a transfer from one agent to another, neither $x\succ y$ nor $y\succ x$. Second consider the problem of a consumer faced with $n$ goods who prefers one bundle to another if it contains at least as much of each good and strictly more of at least one. It is clear that the strong relation is yet again incomplete. Both these examples present what \textit{should} be obvious embeddings into $\mathbb R^n$: the collection of projections onto (in the first case) the utility functions of the respective agents and (in the second case) the respective quantities of the individual goods. Indeed, it is easy to check that both situations would satisfy the conditions of Theorem \ref{conj}. However, observe that the current formulation dictates that if $v(x)v(y)$ for all $v\in\mathcal{V}$ except $w$ for which $w(x)=w(y)$, then $(x,y)\notin Q$. Again this observation violates the economic interpretation. This leads us to define a refined notion of embedding which is compatible with Pareto improvement. In particular, we will slightly modify the typical product order on $\mathbb R^I$. \begin{define}\label{pareto} We define a Pareto order on $\mathbb R^I$ where $(x_\alpha)_{\alpha\in I}\succ(y_\alpha)_{\alpha\in I}$ if $x_\alpha\geq y_\alpha$ for all $\alpha\in I$ and $x_\alpha> y_\alpha$ for some $\alpha\in I$. \end{define} As usual, $x\succeq y$ denotes the negation of $y\succ x$. Clearly $(\mathbb R^I,\succ)$ is asymmetric and transitive and $(\mathbb R^I,\succeq)$ is complete and negatively-transitive as we would like. Now we will speak of (continuous) Pareto embeddings as (continuous) order embeddings of a relation $R$ into $(\mathbb R^I,\succeq)$. In other words, we will be seeking families $\mathcal{V}$ of (continuous) real maps on $X$ such that \begin{align} Q=\{(x,y)\in X\times X:\left\{\begin{array}{ll}v(x)\geq v(y)&\forall v\in\mathcal{V}\\v(x)>v(y)&\exists v\in\mathcal{V}\end{array}\right.\}. \end{align} Theorem \ref{conj} has a non-continuous counterpart for Pareto justifiability. \begin{thm} A relation $P$ is complete and negatively-transitive if and only if it is embeddable in $(\mathbb R^I,\succeq)$. \end{thm} \begin{proof} Let $\{Q_\alpha\}_{\alpha\in I}$ be the asymmetric and negatively-transitive realization of $Q=(P^*)^c$ as implied by Lemma \ref{decomp}. We claim that \begin{align*} Q=\bigcap_{\alpha\in I}Q_\alpha=\left(\bigcup_{\alpha\in I}Q_\alpha\right)\cap\left(\bigcap_{\alpha\in I}P_\alpha\right). \end{align*} The ``$\subset$'' inclusion is trivial. To see the other inclusion, there are three cases: 1) If $(x,y)\in Q$, then $(x,y)\in Q=\bigcap_{\alpha\in I}Q_\alpha$ so the inclusion is tautological. 2) If $(x,y)\notin Q$ and $(y,x)\in Q$, then $(y,x)\in\bigcap_{\alpha\in I}Q_\alpha$, so there is an $\alpha$ such that $(y,x)\in Q_\alpha$. This implies $(x,y)\notin P_\alpha$ so $(x,y)$ is not in the right-hand side. 3) If $(x,y)\notin Q$ and $(y,x)\notin Q$, then the second part of Lemma \ref{decomp} implies that there is an $\alpha$ such that $(y,x)\notin Q_\alpha$. By the same argument as the last case, this implies $(x,y)$ is not in the right-hand side. This proves the ``$\supset$'' direction. Thus $\{Q_\alpha\}_{\alpha\in I}$ satisfies Definition \ref{pareto}. It is well-known that every asymmetric and transitive relation has a utility representation. The collection of such representations over all $\alpha$ in $I$ is evidently a multi-utility representation of $Q$. \end{proof} However, the next example illustrates that this theorem does not have a continuous counterpart along the lines of Theorem \ref{conj}. \begin{ex} Consider the same preferences as Example \ref{sugar}, ie.\ $X=\mathbb{R}$ and \begin{align*} (x,y)\in Q\quad\iff\quad x>y+\epsilon. \end{align*} Suppose for the sake of contradiction that there exists a continuous embedding of $Q$ in $(\mathbb R^I,\succ)$; ie.\ there exists a collection $\mathcal{V}$ of continuous functions $\mathbb{R}\to\mathbb{R}$ such that \begin{align*} Q=\left(\bigcup_{v\in\mathcal{V}}\{v(x)>v(y)\}\right)\cap\left(\bigcap_{v\in\mathcal{V}}\{v(x)\geq v(y)\}\right). \end{align*} The left-hand side is open and, since the functions in $\mathcal{V}$ are continuous, the right-hand side is an intersection of an open set with a closed set. It follows that relative to $\bigcup_{v\in\mathcal{V}}\{v(x)>v(y)\}$, $Q$ is an intersection of closed sets, but at the same time it is clearly open. By a well-known fact from topology, any clopen set must be a (possibly empty) union of connected components of the entire space. But $Q=\{(x,y)\in X\times X:x>y+\epsilon\}$ is a connected subset of $\mathbb{R}^2$. Thus one of the connected components of $\bigcup_{v\in\mathcal{V}}\{v(x)>v(y)\}$ is $Q$. Consider the point $(x,y)=(\epsilon,0)$ which lies on the boundary of but is not an element of $Q$. We cannot have $(\epsilon,0)\in\bigcup_{v\in\mathcal{V}}\{v(x)>v(y)\}$ because then the connected component containing $Q$ would in fact be larger than $Q$. It follows that $(0,\epsilon)\in\bigcap_{v\in\mathcal{V}}\{v(x)\geq v(y)\}$. But it cannot be the case that $(0,\epsilon)\in\bigcap_{v\in\mathcal{V}}\{v(x)=v(y)\}$, because then $0$ and $\epsilon$ would compare identically to all other choices. Thus $(0,\epsilon)\in\bigcap_{v\in\mathcal{V}}\{v(x)\geq v(y)\}\setminus\bigcap_{v\in\mathcal{V}}\{v(x)=v(y)\}\subset\bigcup_{v\in\mathcal{V}}\{v(x)>v(y)\}$. Thus $(0,\epsilon)\in\left(\bigcup_{v\in\mathcal{V}}\{v(x)>v(y)\}\right)\cap\left(\bigcap_{v\in\mathcal{V}}\{v(x)\geq v(y)\}\right)=Q$, which is a contradiction. We conclude that there is no such embedding. \end{ex} It is disappointing that continuous embeddings do not exist in this Pareto sense, but encouraging that it fails only due to a technical topological reason.  \end{document}